\documentclass[12 pt]{article}

\usepackage{amsfonts,amssymb}
\usepackage{color}
\usepackage{amsthm, amssymb, latexsym, amsmath,  
}
\usepackage[margin=1.25 in]{geometry}
\newtheorem{theorem}{Theorem}[section]

\newtheorem{lemma}[theorem]{Lemma}

\newtheorem{proposition}[theorem]{Proposition}
\newtheorem{corollary}[theorem]{Corollary}

\newtheorem{remark}{Remark}
\newtheorem{definition}{Definition}

\newcommand{\DETAILS}[1]{}

\begin{document}

\title{The Hartree-von Neumann limit of many body dynamics\thanks{This paper is a part of the first author's Ph.D
thesis.}
}
\author{ I. Anapolitanos \thanks{Department of Mathematics, University
of Toronto, Toronto, Canada; Supported by NSERC under Grant NA7901.}
\  I.M Sigal \thanks{School of Mathemaitics, IAS, Princeton, N.J.,
U.S.A., on leave from Dept. of Mathematics, Univ. of
Toronto, Toronto, Canada; Supported by NSERC under Grant NA7901.}}

\date{ April 20, 2009}
\maketitle

\abstract{In the mean-field regime, we prove 
convergence (with explicit bounds) of the many-body von Neumann dynamics with bounded
interactions to the Hartree-von Neumann dynamics.}

\section{Introduction}

Derivation of macroscopic equations from microscopic ones is one of
the main challenges of Mathematical Physics. This is usually a
daunting task met with a very limited success. In the last few years
a considerable progress was made on one such problem - derivation of
the Hartree, Hartree-Fock and Gross-Pitaevskii equations in the mean-field and
Gross-Pitaevskii regimes, respectively (\cite{BEGMY, BGGM, FGS, FKS,
FKP, AN, LSY, LS, ErY, ESY, ES, KF, KM, GM, KSS}). Though the work on
the Gross-Pitaevskii limit is quite recent, the work on the
mean-field one goes back to the papers
\cite{He,GV, S}. 

In this note we prove the convergence of solutions of the $N-$body
von Neumann equation with product initial conditions to the $N-$fold product
of solutions of the Hartree-von Neumann equation with the
corresponding initial conditions and estimate the rate of
this convergence. 
The regime we consider is the mean-field one, i.e. with the number
of particles going to infinity, while the strength of interaction
decreasing in the inverse proportion to the number of particles. In
our analysis we follow closely the beautiful work \cite{FGS}. 
One of the new elements of our approach is a Hamiltonian formulation of the Hartree-von Neumann equation. While
this work has been written up there appeared e-prints \cite{RS}, \cite{ErS} and  \cite{GMM}
giving, by 
different techniques, estimates of the rate of convergence in the case of the
$N-$body Schr\"odinger and Hartree equations.

Let $\rho^{\otimes N}:=\rho \otimes ...\otimes \rho$. We start with
the time-dependent von Neumann equation for a system of N bosons
\begin{equation}\label{von Neumann}
\left\{ \begin{array}{ccc} i \hbar \frac{\partial \rho_N} {\partial
t}= [H_N, \rho_N] \\ \rho_N|_{t=0}= \rho_0^{\otimes N},
\end{array} \right.
\end{equation}
with  $\rho_N = \rho_N (t) $ acting on $L^2(\mathbb{R}^{3})^{\otimes N}$ and $\rho_0 $, a
positive, trace-class operator on $L^2(\mathbb{R}^{3})$ of trace $1$. Here $ H_N=H_N^0+V_N $,
with $ H_N^0=-\sum_{i=1}^N h_{x_i},\ h_x \equiv h$ is a self-adjoint operator in a variable $x$, e.g. $h_x:=\frac{\hbar^2}{2} \Delta_x +W(x)$, and
\begin{equation}\label{Nbodpot}
V_N=\frac{g}{2} \sum_{i \neq j}^N  v(x_i-x_j).
\end{equation}
Since $\sum_{i \neq j}^N v(x_i-x_j)=\sum_{i \neq
j}\frac{1}{2}(v(x_i-x_j)+v(x_j-x_i))$ we can assume without loss of
generality that the two-body potential $v$ is even: $v(x)=v(-x)$.
We consider the mean field regime:
 $N \rightarrow \infty$ and $g \rightarrow 0$ with $gN \rightarrow c$. By changing $v$, if necessary, we can assume that
\begin{equation}\label{g1n}
g=\frac{1}{N}.
\end{equation}
We will relate the von Neumann equation \eqref{von
Neumann} to the Hartree-von Neumann equation
\begin{equation}\label{Hartree-von Neumann}
\left\{ \begin{array}{ccc} i \hbar \frac{\partial \rho}{\partial
t}=[h +(v*n_\rho), \rho ] \\
\rho|_{t=0}=\rho_0, \end{array} \right.
\end{equation}
where $v=v(x)$ is the same two-body potential as above and $n_\rho
(x, t):= \rho (x; x, t)$,  the probability or charge
density, with $\rho (x; y, t)$ the integral kernel of $\rho$.

For $v \in L^\infty$ one can show easily that \eqref{Hartree-von Neumann} is globally well-posed on the space of positive, trace-class operators and that the trace, $Tr\rho$, and the energy,
\begin{equation} \label{E}
E(\rho):=Tr
( 
h \rho)+ \frac{1}{2}\int n_\rho\ v * n_\rho,
\end{equation}
 are conserved. Moreover, 
$\rho$ is non-negative, provided so is $\rho_0$. See Appendix \ref{app:HvN}. 

In Section \ref{secCFT} we show that \eqref{Hartree-von Neumann} is a Hamiltonian system with the Hamiltonian \eqref{E} and the Poisson bracket
\begin{equation}\label{poisson}
\{A(\rho), B(\rho)\}=-\frac{i}{\hbar}Tr
\left(\partial_{\rho}A(\rho)\rho
\partial_{\rho} B(\rho)-\partial_{\rho}B(\rho)\rho \partial_{\rho}A(\rho)\right),
\end{equation}
where $A(\rho)$ and  $B(\rho)$ are differentiable functionals of $\rho$ and the operator (Fr\'echet
derivative) $\partial_{\rho}A(\rho)$ is defined by the equation
$Tr(\partial_{\rho}A(\rho)\xi)= \partial_s A(\rho+
s\xi)|_{s=0}.$ This, as was mentioned above, plays an important role in our analysis.  

Finaly, note that 
since the integral kernel of the operator 
$\rho_N|_{t=0} 
$ is symmetric with respect to permutations of particle
coordinates, the same is true for the solution to \eqref{von
Neumann}.

To formulate the main result we need some notation and definitions.
For any Banach space $X$, we denote the space of bounded linear
operators from $X$ to itself by $B(X)$. Let $L_S^2(\mathbb{R}^{3M})$
be the subspace of $L^2(\mathbb{R}^{3M})$ consisting of the
functions that are symmetric with respect to permutation of
particle coordinates, and let
\begin{equation}
P_S^M \Psi(x_1,...,x_M):=\frac{1}{M!}\sum_{\sigma \in S^M}
\Psi(x_{\sigma(1)},...,x_{\sigma(M)}),
\end{equation}
where $S^M$ denotes the permutation group of the set
$\{1,2,...,M\}$, be the orthogonal projection onto the subspace
$L_S^2(\mathbb{R}^{3M})$ of $L^2(\mathbb{R}^{3M})$
 We denote by $I^q$ the identity operator acting on $q$ coordinates.
\begin{definition}
Let $\mathcal{A}_p:=B (L_S^2(\mathbb{R}^{3p}))$. For $p<N$, we
define the maps $\phi_{p}^N:\mathcal{A}_p \rightarrow \mathcal{A}_N$
by
\begin{equation}\label{lifting}
\phi_p^N(a):=P_S^N(a \otimes I^{N-p})P_S^N.
\end{equation}
\end{definition}
Let $\mathcal{A}_{N,p}$ be the image of $\mathcal{A}_p$ under
$\phi_p^N$. Its elements will be called \textit{quantum $p$-particle
observables}
or simply \textit{$p$-particle observables}. Note that
$\mathcal{A}_{N,1} \subset \mathcal{A}_{N,2} \subset  ...\
\mathcal{A}_{N,N}=\mathcal{A}_N$.

The following theorem relates the asymptotic behavior of \eqref{von
Neumann} as $N \rightarrow \infty$ with \eqref{Hartree-von Neumann}.
 \begin{theorem}\label{main}
Assume that $v$ is bounded and 
even and that \eqref{g1n} holds. Let $\rho_N$ solve \eqref{von
Neumann} with $Tr \rho_0=1$. Then for any $p$ and any $A=\phi_p^N(a) \in
\mathcal{A}_{N,p},$ and for  $N \rightarrow \infty,$ $$  Tr (A\rho_N)-Tr( A \rho^{\otimes N}) \rightarrow 0,$$ where $\rho$ solves \eqref{Hartree-von Neumann}. Moreover, we have the estimate
\begin{equation}\label{2star}
|Tr (A\rho_N)-Tr( A \rho^{\otimes N})| \leq
2^{([\frac{t}{\tau}]+2)p} N^{-\gamma (t)}\|a\|_{\mathcal{A}_{p}},
\end{equation}
where
$\tau=\frac{\hbar}{8\|v\|_{\infty}}$, $[\frac{t}{\tau}]$ is the
integer part of $\frac{t}{\tau}$ and $\gamma (t):=\frac{ 1}{4e
[\frac{t}{\tau}+1]!}$.
 \end{theorem}
\begin{remark}
The $N$-bound in this theorem is rather poor, especially compared with estimates for the Schr\"odinger equation mentioned above. It is improved
somewhat in \cite{An}, where an extension of this result to Coulomb-type potentials is also presented.
\end{remark}

From now on we fix the particle number $N$ and sometimes drop the
corresponding index from the notation. For instance, we write $P_S$
for $P_S^N$, $\phi_p$ for $\phi_p^N,$ and $V$ for $V_N$. 

The rest of the paper is devoted to a proof of Theorem \ref{main}.
The paper is organized as follows. In Section \ref{secG} 
we derive a convenient equation for the map $
\Gamma_t(A):=e^{\frac{i H_N t}{\hbar}} e^{\frac{-i H_N^0 t}{\hbar}}
A e^{\frac{i H_N^0 t}{\hbar}} e^{\frac{-i H_N
 t}{\hbar}}$, which is connected to the l.h.s. of \eqref{2star}. To this end we use a decomposition of the commutator
with the many-body potential into the tree and loop operators,
introduced in \cite{FGS}. We use this equation in Section
\ref{secAG} in order to approximate $\Gamma_t$ by an operator
$\Gamma_t^H$ whose expansion contains only tree operators. In
Section \ref{secCFT} we discuss the Hamiltonian and Liouvillean
formulations of the Hartree-von Neumann equation and the Dyson
expansion for the latter.
We also show that in certain (symbolic)
representation the tree operators act as Poisson brackets. 
In Section \ref{secCtau} we prove the result of Theorem \ref{main}
for small times and in Section \ref{secC} we use the group
properties of the von Neumann and Hartree-von Neumann dynamics in
order to extend the proof to all times.

\section*{Acknowledgement} The second author (I.M.S.) is grateful to J\"urg Fr\"ohlich for numerous stimulating discussions and for introducing him to the mean-field problems, and
to IAS and ETH-Z\"urich, for hospitality. J\"urg Fr\"ohlich has
informed one of us (I.M.S.) that he has defined a different Poisson bracket for
the Hartree-von Neumann equation.

 \section{Map $\Gamma_t$ and its equation}\label{secG}

In this section we derive a convenient equation for the family of
operators $\Gamma_t$ acting on elements of
 $\mathcal{A}_N$ and defined by
\begin{equation}\label{subtractfree}
 \Gamma_t(A):=e^{\frac{i H_N t}{\hbar}} e^{\frac{-i H_N^0 t}{\hbar}} A e^{\frac{i H_N^0 t}{\hbar}} e^{\frac{-i H_N
 t}{\hbar}}.
\end{equation}
This family is related  to the l.h.s. of \eqref{2star} as $ Tr (A
\rho_N) = Tr( \Gamma_t(A_t) \rho_0^{\otimes N})$, where
$\rho_0,\rho_N$ are the same as in \eqref{von Neumann} and $A_t$
denotes the free evolution of $A$:
 \begin{equation}\label{free}
 A_t:=e^{\frac{i H_N^0 t}{\hbar}} A e^{\frac{-i H_N^0 t}{\hbar}}.
\end{equation}
Writting  $\Gamma_r(A)$ as the integral of derivative, we obtain
\DETAILS{By differentiating $\Gamma_r(A)$ with respect to $r$ and
using that operators commute with their exponentials we obtain that
$\partial_r \Gamma_r(A)=\frac{i}{h} \Gamma_r([V_r,A])$, which after
integration yields the equation}
\begin{equation}\label{derivative}
\Gamma_t(A)=A+\frac{i}{\hbar} \int_0^t \Gamma_r ([V_r,A])dr.
\end{equation}

A simple analysis shows that the solution of this equation is unique
in $C([0,\infty),B(\mathcal{A}_N)).$

Now we decompose the commutator on the r.h.s. of \eqref{derivative}
in a convenient way (cf. \cite{FGS}). For any $i,j \leq N, i \neq j$ denote by
$V^{ij}$ the multiplication operator by $v(x_i-x_j)$. 

\begin{proposition}
We have for $a \in \mathcal{A}_p$,
\begin{equation}\label{treeloop}
\frac{i}{\hbar}[V_r,\phi_p(a)]=T_r(\phi_p(a))+L_r(\phi_p(a)),
\end{equation}
where $T_r(\phi_p(a))=0$ for $p \ge N$ and otherwise
\begin{equation}\label{newtree}
T_r(\phi_p(a))= \frac{N-p}{N} \phi_{p+1}(X_{p,r}(a)),
\end{equation}
with $X_{p,r}: \mathcal{A}_p \rightarrow \mathcal{A}_{p+1}$, defined
by
\begin{equation}\label{mta}
X_{p,r}(a):= p P_S^{p+1} \frac{i}{\hbar} [V_r^{p,p+1},a \otimes
I^1]P_S^{p+1},
\end{equation}
for $p < N$, and
\begin{equation}\label{newloop}
L_r(\phi_p(a))= \phi_p(Y_{p,r}(a)), \ \mbox{with}\
Y_{p,r}(a)=\frac{p(p-1)}{2N} (P_S^p \frac{i}{\hbar}[V_r^{p,p-1},a]
P_S^p ).
\end{equation}
\end{proposition}
\begin{proof}
 Recall the notation $V_t=e^{\frac{iH_N^0 t}{\hbar}} V e^{-\frac{iH_N^0 t}{\hbar}}$. Using equations
 $V_r= \frac{1}{N} \sum_{i<j}^{1,N}V_r^{ij}$ and \eqref{g1n}, we obtain
\begin{equation}\label{decomp}
[V_r,A]=\frac{1}{N} \sum_{i<j}^N [V_r^{ij},A].
\end{equation}
Since the operator $V_r= \frac{1}{N} \sum_{i<j}^{1,N}V_r^{ij}$ is
permutationally symmetric, we have from \eqref{lifting} that
$$\frac{i}{\hbar}[V_r,A]=\frac{i}{\hbar}P_S[V_r,a_{} \otimes
I^{N-p}]P_S$$
\begin{equation}
= \frac{i}{\hbar N} P_S\sum_{i<j}^{1,N}[V_r^{ij},a_{} \otimes
I^{N-p}])P_S= T_r(A)+ L_r(A),
\end{equation}
where
\begin{equation}\label{treealt}
T_r(A)
=\frac{i}{N\hbar}P_S(\sum_{i=1}^p \sum_{j=p+1}^N [V_r^{ij},a\otimes
I^{N-p}])P_S
\end{equation}
and
$$L_r(A)
=\frac{1}{N} P_S(\frac{i}{\hbar}\sum_{
i<j}^{1,p}[V_r^{ij},a \otimes I^{N-p}])P_S.$$ By symmetry we have
$$T_r(A)=\frac{p(N-p) }{N } P_S(\frac{i }{\hbar}[V_r^{p,p+1},a\otimes
I^{N-p}])P_S.$$
Now, since $ P_S^{p+1}P_S=P_S P_S^{p+1}=P_S$, we have furthermore
$$T_r(A)=
\frac{N-p}{N} P_S(p P_S^{p+1}(\frac{i}{\hbar}[V_r^{p,p+1},a\otimes
I])P_S^{p+1} \otimes I^{N-p-1}) P_S$$
$$\stackrel{\eqref{lifting}}{=} \frac{N-p}{N} \phi_{p+1}(p P_S^{p+1}
\frac{i}{\hbar}[V_r^{p,p+1},a \otimes I]P_S^{p+1}),$$ which gives
\eqref{newtree}-\eqref{mta}. Similarly,
we find
$$L_r(A)= \frac{p(p-1)}{2N}  P_S(\frac{i}{\hbar}[V_r^{p,p-1},a]\otimes I^{N-p})P_S.$$
This, due to \eqref{lifting}, gives \eqref{newloop}. \end{proof}

\begin{remark}
In general, $T_r(\phi_p(a))\neq
T_r(\phi_{q}(b))$, even if $\phi_p(a)=\phi_{q}(b)$. Thus, 
e.g. the expression $T_r A$ should be understood as $T_r\phi_p(a)$ with $A=\phi_p(a)$.
However, our abuse of notation will not cause a confusion. \DETAILS{If
$A=\phi_p(a)$, then $A=\phi_{p+1}(P_S^{p+1}(a \otimes I)
P_S^{p+1}).$ A simple analysis can show that
\begin{equation}
T_r(\phi_p(a))\neq T_r(\phi_{p+1}(P_S^{p+1}(a \otimes I)
P_S^{p+1})).
\end{equation}
In other words if a $p$-particle observable is considered a
$p+1$-particle observable, then the tree term of the commutator
changes. Therefore it is important to emphasize that $A$ is a
$p$-particle observable.}
\end{remark}

$T_r$ and $L_r$ will be called the \textit{tree and loop operators},
respectively (see \cite{FGS}). Observe that the equation
\eqref{newtree} implies that
\begin{equation}\label{pplus1}
A \in \mathcal{A}_{N,p} \implies T_r(A) \in \mathcal{A}_{N,p+1}
\end{equation}
and equation \eqref{newloop} implies that
\begin{equation}
A \in \mathcal{A}_{N,p} \implies L_r(A) \in \mathcal{A}_{N,p}.
\end{equation}

Combining equations \eqref{derivative} and \eqref{treeloop} from
above we obtain the following equation for $\Gamma_t$
\begin{equation}\label{inteq}
\Gamma_t(A)=A+ \int_0^t \Gamma_{r}(T_{r}(A))dr+ \int_0^t
\Gamma_{r}(L_{r}(A)))dr.
\end{equation}
Introducing the notation $\Gamma \equiv \Gamma_.=(\Gamma_t, t \geq
0)$ we rewrite this equation in a more compact way
\begin{equation}\label{compactway}
\Gamma=I+K \Gamma+R(\Gamma)
\end{equation}
where
\begin{equation}\label{K}
(K G)_t:= \int_0^t G_s T_s ds
\end{equation}
and
\begin{equation} \label{RG}
 R(\Gamma)_t= \int_0^t \Gamma_s L_s ds.
\end{equation}
Equation \eqref{K} defines the operator $K:G \rightarrow K G $
acting on the families $G=\{G_t \in B(\mathcal{A}_N),t \geq 0\}$.
Clearly, $K$ is a bounded operator on $C([0,T],B(\mathcal{A}_N)),
\forall T \geq 0.$

\begin{proposition}\label{Kinv}
Let $K$ be the operator defined in equation \eqref{K}. Then $I-K$
is invertible
and
\begin{equation} \label{1-Kinv}
\forall A\in \mathcal{A}_{N,p},\  (I-K)^{-1}GA=\sum_{n=1}^{N-p}K^n
GA.
\end{equation}
\end{proposition}
\begin{proof} Let $A\in \mathcal{A}_{N,p}$.
The definition of the operator $K$ implies that
\begin{equation} \label{Kn}
(K^n G)_tA=\int_{\Delta^t_n}d^{n}t G_{t_n} T_{t_n}...T_{t_1}A,
\end{equation}
where $\int_{\Delta^t_n}d^nt=\int_0^t
dt_1\int_0^{t_1}dt_2...\int_0^{t_{n-1}}dt_n$. (Here $\Delta^t_n$ is
the $n$-symplex $0 \le t_n \le t_{n-1} \le ... \le t_1 \le t$.)

 Equations \eqref{pplus1} and $T_r(A)=0 \quad \forall A \in \mathcal{A}_{N,N}$ imply
\begin{equation}\label{treezero}
T_{r_n} T_{r_{n-1}} ... T_{r_1} A=0, \quad \forall A \in
\mathcal{A}_{N,p},\quad n > N-p.
\end{equation}
Hence, by \eqref{treezero}, $(K^n G)_tA =0$ for $n > N-p$. This
gives \begin{equation} \label{series}\sum_{n=1}^{\infty}K^n GA =
\sum_{n=1}^{N-p}K^n GA.
\end{equation} On the other hand,
$(I-K)\sum_{n=1}^{\infty}K^n GA =GA$ and $\sum_{n=1}^{\infty}K^n
(I-K)GA=GA,$ which completes the proof.
\end{proof}

\section{Approximation of $\Gamma$}\label{secAG}
Let $\Gamma^{(H)}:=\sum_{n=0}^{\infty} K^n I$, which, due to
\eqref{series}, is a finite series on $\mathcal{A}_{N,p}$.
Equivalently, we write
\begin{equation}\label{GammaH}
 \Gamma_t^{(H)}:=\sum_{n=0}^{\infty} \int_{\Delta^t_{n}}d^{n}t\ T_{t_n}...T_{t_1}.
\end{equation}
\DETAILS{Recall the notation
\begin{equation}\label{tau}
\tau:=\frac{\hbar}{8\|v\|_{\infty}}.
\end{equation}}
\begin{proposition}
For $t \leq \tau:=\frac{\hbar}{8\|v\|_{\infty}}$, we have
\begin{equation}\label{flowconv}
\|(\Gamma_t-\Gamma_t^H)\phi_p\|_{\mathcal{A}_{p} \rightarrow
\mathcal{A}_N} \leq 2^{p-2} \frac{p}{N} \frac{t}{\tau}.
\end{equation}
\end{proposition}
\begin{proof}
 Using Proposition \ref{Kinv} and equation
\eqref{compactway} we obtain that
\begin{equation}\label{Gdiff}
\Gamma-\Gamma^H=(I-K)^{-1}R(\Gamma)=\sum_{n=0}^{\infty}K^n
R(\Gamma).
\end{equation}
Using equations \eqref{RG} and \eqref{Kn}, we find that
 $$(K^n R (\Gamma))_t = \int_{\Delta^t_{n+1}}d^{n+1}t\ \Gamma_{t_{n+1}}L_{t_{n+1}}
T_{t_n}...T_{t_1}.$$ Using equation \eqref{newtree} and
\eqref{newloop}, we obtain that
\begin{equation}\label{mutlitreeloop}
L_{t_{n+1}}T_{t_n}...T_{t_1}
\phi_p(a)$$$$=\phi_{p+n}(\frac{(N-p)!}{(N-p-n)!N^n}Y_{p+n,t_{n+1}}X_{p+n-1,t_n}X_{p+n-2,t_{n-1}}...X_{p,t_1}(a)).
\end{equation}
Using equations \eqref{mta}, $\|P_S^M\|=1$ and
$\|V^{ij}\|_{B(L^2(\mathbb{R}^6))}= \|v\|_{\infty}$,
we derive the estimate on the tree and loop operators 
\begin{equation}\label{mtaest}
\|X_{p,t}(a)\|_{\mathcal{A}_{p+1}}\leq \frac{2
\|v\|_{\infty}}{\hbar} p \|a\|_{\mathcal{A}_p},
\end{equation}
\begin{equation}\label{loopnorm}||Y_{p,t}(a)||_{ \mathcal{A}_{p}} \leq
\frac{\|v\|_{\infty}}{\hbar}\frac{p(p-1)}{N}\|a
\|_{\mathcal{A}_{p}}.\end{equation} Equations \eqref{mtaest},
\eqref{loopnorm}, \eqref{mutlitreeloop} and
$\|\phi_p(a)\|_{\mathcal{A}_{N,p}} \le \|a\|_{\mathcal{A}_p}$
(since $\|P_S\|=1$) imply that
\begin{equation}
\|L_{t_{n+1}}T_{t_n}...T_{t_1} \phi_p\|_{\mathcal{A}_p \rightarrow
\mathcal{A}_{N,p+n}} $$$$ \leq \left(\frac{2
\|v\|_{\infty}}{\hbar}\right)^n \frac{(N-p)!}{(N-p-n)!N^n}
\frac{(p+n)!}{(p-1)!}\frac{(p+n-1)}{2N}\|a\|_{\mathcal{A}_p},
\end{equation}
which together with the fact that
$\|\Gamma_t\|_{\mathcal{A}_{N,p}\rightarrow \mathcal{A}_N} = 1$ and
the equality
\begin{equation}\label{Vo}
\int_{\Delta^t_{n}}d^{n}t\ \equiv \int_0^t dt_1
\int_0^{t_1}dt_2...\int_0^{t_{n-1}} dt_n=\frac{t^n}{n!},
\end{equation}
implies that
\begin{equation}\label{Knest}
||(K^n R (\Gamma))_t \phi_p||_{\mathcal{A}_{p} \rightarrow
\mathcal{A}_N}$$$$ \leq \frac{1}{2}
\left(\frac{t}{4\tau}\right)^{n+1} \frac{(N-p)!}{(N-p-n)!N^{n}}
\frac{(p+n)!}{(p-1)!(n+1)!} \frac{p+n-1}{N}.
\end{equation}
Furthermore, using the inequalities
\begin{equation}\label{2nbnd}\frac{(p+n)!}{(p-1)!(n+1)!} \leq
2^{p+n}
\end{equation}
and
\begin{equation}\label{N-pbnd}\frac{(N-p)!}{(N-p-n)!N^n} \leq 1,
\end{equation}
we simplify \eqref{Knest} as
\begin{equation}\label{Knr}
\|(K^n R (\Gamma))_t \phi_p\|_{\mathcal{A}_{N,p} \rightarrow
\mathcal{A}_N} \le 2^{p-2}\left(\frac{t}{2 \tau}\right)^{n+1}
 \frac{p+n-1}{N}.
\end{equation}
To conclude our calculations we use the following equality:
\begin{equation}\label{numberseq1}
\sum_{n=0}^{\infty} 2^{-n} (p+n-1) =2p.
\end{equation}
To derive \eqref{numberseq1}, let
$A=\sum_{n=0}^{\infty}2^{-n}(p+n-1).$
 Then
$\frac{A}{2}=\sum_{n=0}^{\infty}2^{-n-1}(p+n-1)=\sum_{n=1}^{\infty}2^{-n}(p+n-2),$
and therefore subtracting the latter equation 
from the former, we obtain that
$\frac{A}{2}=(p-1)+\sum_{n=1}^{\infty}2^{-n}=p$, which gives
immediately equation \eqref{numberseq1}.
 Equation \eqref{numberseq1}
 yields for $t\leq \tau$ that
\begin{equation}
\sum_{n=0}^{N-p-1} \|(K^n R (\Gamma))_t \phi_p\|_{\mathcal{A}_{N,p}
\rightarrow \mathcal{A}_N} \leq 2^{p-2}\frac{p}{N}
\frac{t}{\tau}.
\end{equation}
This together with \eqref{Gdiff} and \eqref{series} proves equation
\eqref{flowconv} for $t\leq \tau$.
\end{proof}

\section{Classical Field Theory for the Hartree-von Neumann Equation}\label{secCFT}
 In this section we develop Hamiltonian and Liouvillian
representations for the Hartree equation viewed as a classical field
theory. We begin with key definitions. Let $\mathbb{I}_1$ denote the
space of all positive, trace class operators on $L^2(\mathbb{R}^3)$.
 For any 
operator $ a \in \mathcal{A}_{p}$ 
we define the $p$-particle classical field
observable $A^c: \mathbb{I}_1 \rightarrow \mathbb{C}$ by the
equation
\begin{equation}\label{clobs}
 A^c(\rho):=  
Tr(a \rho^{\otimes p}).
\end{equation}
We equip the space of these functionals with the norm
 \begin{equation}\label{clnorm}
||A^c||:=sup_{\rho \in \mathbb{I}_1, Tr\rho =1}|A^c(\rho)|.
\end{equation}
The last two equations imply that
\begin{equation}\label{claqua}
\|A^c\| \leq \|a\|_{\mathcal{A}_p}.
\end{equation}

For a functional $A^c(\rho)$ we define the operator (Fr\'echet
derivative) $\partial_{\rho}A^c(\rho)$ by the equation
$Tr(\partial_{\rho}A^c(\rho)\xi)= \partial_s A^c(\rho+
s\xi)|_{s=0}.$ On the space of classical field observables we define
the Poisson bracket by
\begin{equation}\label{poisson}
\{A^c(\rho), B^c(\rho)\}=-\frac{i}{\hbar}Tr
\left(\partial_{\rho}A^c(\rho)\rho
\partial_{\rho} B^c(\rho)-\partial_{\rho}B^c(\rho)\rho \partial_{\rho}A^c(\rho)\right).
\end{equation}
The Jacobi identity is proven in Appendix \ref{app:Jacobi}. Note that if $A^c$ is a $p$-particle classical field observable, and
$B^c$ is a $q$-particle classical field observable, then
$\{A^c,B^c\}$ is a $p+q-1$-particle classical field observable.

Furthermore, we observe that
\begin{equation}\label{poissonRhoPsi}
\{A^c(\rho),
B^c(\rho)\}|_{\rho=P_{\psi}}=\frac{i}{\hbar}\int(\partial_{\psi(x)}A^c\partial_{\overline{\psi}(x)}
B^c-\partial_{\overline{\psi}(x)}A^c\partial_{\psi(x)}B^c)(\psi,\overline{\psi})dx,
\end{equation}
where $P_{\psi}$ is the rank-one projection on the vector $\psi$ and
$A^c(\psi,\overline{\psi}):= A^c(P_{\psi})$. The r.h.s. is the
standard Poisson bracket for the Hartree equation (see \cite{FGS}).
Indeed, $\partial_{\rho}A^c(\rho)$ and $\partial_{\rho}B^c(\rho)$
are operators on $L^2(\mathbb{R}^3)$ ($1-$particle observables) and
therefore
$$Tr((\partial_{\rho}B^c)(P_{\psi})P_{\psi}(\partial_{\rho}A^c)(P_{\psi}))
= \langle (\partial_{\rho}A^c)(P_{\psi})^* \psi,
(\partial_{\rho}B^c)(P_{\psi})\psi \rangle .$$ Since, as it is easy
to see, $(\partial_{\rho}B^c)(P_{\psi}) \psi =
\partial_{\overline{\psi}(x)} B^c(\psi,\overline{\psi})$ and $\overline{(\partial_{\rho}A^c)(P_{\psi})^* \psi} =
\partial_{\psi(x)} A^c(\psi,\overline{\psi})$, this gives
$Tr((\partial_{\rho}B^c)(P_{\psi})P_{\psi}(\partial_{\rho}A^c)(P_{\psi}))
=  \int\partial_{\psi(x)}A^c\partial_{\overline{\psi}(x)} B^c,$
which implies the desired relation.

 Introduce the classical Hamiltonian functional
$$\qquad H^{c}(\rho):=H^{0,c}(\rho)
+V^{c}(\rho),$$ where  (extending \eqref{clobs} to the unbounded 1-particle observable $h$)
\begin{equation} \label{H0c}
H^{0,c}(\rho):=Tr
( 
h \rho),\ \mbox{and}\
V^{c}(\rho):=\frac{1}{2}Tr(v \rho^{\otimes 2}).\end{equation}
Note that the Hartree-von Neumann equation \eqref{Hartree-von
Neumann} is equivalent to the equation \begin{equation} \label{HvN}
\partial_t \rho = \{H^{c}(\rho),
\rho\},
\end{equation} which motivates the above definition of the Poisson bracket.
\begin{remark}
The last equation holds in the weak sense that for all $ a \in \mathcal{A}_{1}$,
\begin{equation}\label{weak1}
\partial_t Tr(a \rho )=Tr(a\ \{H^c, \rho\} ). 
\end{equation}
Using the linearity of the Poisson brackets in the second factor, we
obtain
\begin{equation}\label{weak2}
Tr(a \{H^c, \rho\}) = \{H^c, Tr(a \rho )\}.
\end{equation}
Thus equation \eqref{HvN} is equivalent to the equation
\begin{equation}
\partial_t Tr(a \rho )=\{H^c, Tr(a \rho )\},
\end{equation}
or $\partial_t a^c(\rho)=\{H^c(\rho),a^c(\rho)\}$ for all 1
particle observables $a$. 
\end{remark}

 To show \eqref{HvN} we use the definition of the
Poisson bracket and the relation $Tr(\partial_{\rho}\rho\xi)=
\xi$, which follows from the
definition of $\partial_{\rho}A^c(\rho)$ above to obtain
\begin{equation} \label{PBD} \{H^{c}, \rho\}=
-\frac{i}{\hbar}\left(\partial_{\rho}H^{c}(\rho)\rho-\rho
\partial_{\rho}H^{c}(\rho) \right).
\end{equation} Next, computing
$\partial_{\rho}H^{c}(\rho)$, 
 we conclude that $\{H^c,\rho\}$ is
equal to the r.h.s. of \eqref{Hartree-von Neumann}.

Let $\Phi_t$ be the flow given by the Hartree-von Neumann initial
value problem \eqref{Hartree-von Neumann}, i.e.
$\Phi_t(\rho_0):=\rho_t$ where $\rho_t$ is the solution of
\eqref{Hartree-von Neumann} at time $t$. We denote by $\Phi_t^0$ the
flow of \eqref{Hartree-von Neumann} for $v=0$ (the free flow). We
define the Hartree-von Neumann evolution on the space of
$p$-particle classical field observables by
\begin{equation}\label{clevol}
U_t(A^c(\rho)):=A^c(\Phi_t(\rho)),
\end{equation}
and the free Hartree-von Neumann evolution by
$U_t^0(A^c(\rho)):=A^c(\Phi_t^0(\rho)).$
In a standard way we derive the following equation
\begin{equation}\label{Har}
\partial_tU_t(A^c(\rho))=
U_t(\{H^{cl},A^c\}(\rho))
\end{equation}
(and similarly for $U_t^0(A^c(\rho))$).
\DETAILS{\begin{equation}\label{Harfree}
\partial_tU_t^0(A^c(\rho))=
U_t^0(\{H^{0,cl},A^c\}(\rho)).
\end{equation}}
\DETAILS{\begin{proof} Equation \eqref{Harfree} is immediate
consequence of equation \eqref{Har} by setting $v=0$. The proof of
equation \eqref{Har} follows from the equation
$\partial_t \Phi_t(\rho)= \{H^{cl}(\Phi_t(\rho)),\Phi_t(\rho)\},$
which follows from \eqref{HvN},\eqref{PBD} the definition of
$A^c(\rho)$ and standard properties of the Poisson bracket.
\end{proof}}

Let $V_t^{c}$ denote the free evolution, $U_t^0(V^{c})$, of the
2-particle classical observable $V^{c}$. A simple computation gives
that
\begin{equation}\label{clpot}
V_t^{c}(\rho):=\frac{1}{2} Tr (v_t \rho^{\otimes 2}),
\end{equation}
where $v_t$ is the operator $\psi (x_1, x_2) \rightarrow
e^{\frac{i}{\hbar}(h_{x_1}+h_{x_2})t}v(x_1-x_2)
e^{-\frac{i}{\hbar}(h_{x_1}+h_{x_2}) t}\psi (x_1,
x_2)$.

In what follows we denote the action of Poisson bracket as
\begin{equation}\label{Pvac}
P_V A^c:= \{V,A^c\}.
\end{equation}

\begin{proposition}
We have the following expansion:
\begin{equation}\label{Hartreeflow}
U_t(A^c)=\sum_{n=0}^{\infty}A^c_{t,n},
\end{equation}
where $A^c_{t,0}=A_t^c:=U_t^0(A^c)$  and, for $n \ge 1$,
\begin{equation}\label{atn}
A^c_{t,n}=\int_{\Delta^t_{n}}d^{n}t\
P_{V_{t_n}^{c}}...P_{V_{t_1}^{c}}A_t^c,
\end{equation}
with the following the estimates
\begin{equation}\label{Voltera}
||A_{t,n}^c|| \leq
\left(\frac{t}{2 \tau}\right)^n 2^{p-1} \|a\|_{\mathcal{A}_{p}}
\end{equation}
(in particular, for $t \le \tau$ the series converges in the norm \eqref{clnorm}).
\end{proposition}
\begin{proof}
Define $\tilde{A}_t^c:=U_{-t}^0(U_{t}(A^c))$. By a standard argument
we have that
$\partial_t \tilde{A}_t^c= P_{V_{-t}^{c}} \tilde{A_t^c}.$
\DETAILS{Indeed, using $H^c=H^c_0 + V^c$,  the Leibnitz rule and the
fact that $U_t^0$ is a linear operator, we obtain that
 $$\partial_t \tilde{A}_t^c=\partial_{t} U_{-t}^0 (U_t (A^c)) =
(\partial_t U_{-t}^0) (U_t (A^c))+ U_{-t}^0 (\partial_t U_t (A^c))$$
$$=-
U_{-t}^0(\{H^{0,c},U_t(A^c)\})+
U_{-t}^0 (U_t(\{H^{c},A^c\}))$$
$$=
[-U_{-t}^0 (\{H^{0,c},U_t(A^c)\})+U_{-t}^0(\{H^{c},U_t (A^c)\})]$$
$$=
U_{-t}^0 (\{V^{c},U_t(A^c)\})=
\{V_{-t}^{c},\tilde{A}_t^c\} \stackrel{\eqref{Pvac}}{=}
P_{V_{-t}^{c}}\tilde{A}_t^c.$$
In the third equality we used Lemma \ref{Brackder}, and in the
fourth one, the fact that $U_{-t}^0(H^{0,cl})=H^{0,cl}$ (invariance
of the free Hamiltonian under the free flow). Also we have used the
fact that $U_t(\{A^c,B^c\})=\{U_t(A^c),U_t(B^c))\}$ for any
classical field observables $A^c$ and $B^c$ (respectively for
$U_{-t}^0$).} Integrating this equation, we obtain immediately that
 \begin{equation}
\tilde{A}_t^c=A^c+ \int_0^t dt_1\ P_{V_{-t_1}^{c}}
\tilde{A}_{t_1}^c.
\end{equation}
 Iterating this  equation and applying $U_t^0$ to the result
 we obtain \eqref{Hartreeflow},
with $A^{c}_{t,0}:=U_t^0(A^c)$ and
$A_{t,n}^c:= \int_0^t dt_1 \int_0^{t_1}dt_2 \int_0^{t_{n-1}} dt_n
U_t^0(P_{V_{-t_1}^{c}}...P_{V_{-t_n}^{c}}A^c),\ n \ge 1,$
which after a change of variables of integration gives \eqref{atn}.

It remains to prove \eqref{Voltera}, which shows that the series \eqref{Hartreeflow} converges and
which we are going to prove next.
Recall the
notation $\rho^{\otimes N}=\rho \otimes ... \otimes \rho,$ the $\ N-$fold tensor product. 
 \begin{lemma} Let $X_{m,t}$ be the operator defined in equation
 \eqref{mta}. Then
 \begin{equation}\label{Bracketcommut1}
Tr ( X_{p,t}(a)\rho^{\otimes p+1})= P_{V_{t}^{c}} a (\rho).
\end{equation}
\end{lemma}
\begin{proof} \DETAILS{Using \eqref{clpot} and the fact that 
$v$ is even (which implies symmetry with respect to permutations of
the particle coordinates) we have that
$$Tr \partial_{\rho} A^c \rho \partial_{\rho}
V_t^{cl} $$ $$=p \int a(x_1,...x_p;y_1,...,y_{p})\rho(y_1;
x_1)...\rho(y_{p-1}; x_{p-1})\rho(y_{p}; z_{p}) $$ $$\times
v_t(z_{p},x_{p+1}; x_p,y_{p+1}) \rho(y_{p+1}; x_{p+1})
dx_{1,p+1}dy_{1,p+1} dz_p,$$ where $dr_{k,l}$ denotes $dr_k
dr_{k+1}...dr_l$. Relabeling the variables of integration as $z_p
\rightarrow x_p$ and $x_p \rightarrow z_p$, we obtain
\begin{equation}\label{firstcomterm}
Tr \partial_{\rho} A^c \rho \partial_{\rho} V_t^{cl}=p
(V_t^{p,p+1}(a\otimes I) )^c .
\end{equation}
Similarly, we have that
\begin{equation}\label{secondcomterm}
Tr \partial_{\rho} V_t^{cl}\rho \partial_{\rho} A^c  =p((a \otimes
I) V_t^{p,p+1})^c.
\end{equation}
Equations \eqref{mta},\eqref{poisson},\eqref{Pvac},
\eqref{firstcomterm},
   and \eqref{secondcomterm} give equation
\eqref{Bracketcommut1}.}

The proof follows from the relation $\partial_{\rho}A^c(\rho)=pTr_{p-1}(a \rho^{p-1}\otimes I)$,
where $Tr_{p-1}$ is the partial trace over the first $p-1$
coordinates, and a simple computation. 
\end{proof}

Iterating \eqref{Bracketcommut1} 
we obtain equation
\begin{equation}\label{multibrac}
Tr (  X_{p+n-1,t_n}...X_{p,t_1}(a) \rho^{\otimes p+n})=
P_{V_{t_n}^{c}}...P_{V_{t_1}^{c}}a (\rho).
\end{equation}
%
%
%
\DETAILS{\begin{lemma} For any $p$-particle observable $A=\phi_p(a)$
we have that
\begin{equation}\label{Voltera}
||A_{t,n}^c|| \leq
\left(\frac{t}{2 \tau}\right)^n 2^{p-1} \|a\|_{\mathcal{A}_{p}}.
\end{equation}
\end{lemma}
\begin{proof}} Next, using equation \eqref{mtaest} we obtain that
 \begin{equation}\label{mulcoest}
\|X_{p+n-1,t_n}...X_{p,t_1}(a)\|_{\mathcal{A}_{p+n}} \leq \left(
\frac{2\|v\|_{\infty}}{\hbar} \right)^n \frac{(p+n-1)!} {(p-1)!}
\|a\|_{\mathcal{A}_p}.
 \end{equation}
The last two equations together with 
\eqref{clnorm}, the formula $A_t=\phi_p(a_t)$, where  $a_t=e^{\frac{i H_p^0
t}{\hbar}} a e^{-\frac{i H_p^0 t}{\hbar}}$, 
and the isometry of the free evolution, $a_t$,
 \DETAILS{$$|P_{V_{t_n}^{cl}}...P_{V_{t_1}^{cl}}A_t^c(\rho)|
 \stackrel{\eqref{multibrac}}{=}|Tr (   X_{p+n-1,t_n}...X_{p,t_1}(a_t) \rho^{\otimes (p+n)})| $$
 $$
 \le \|X_{p+n-1,t_n}...X_{p,t_1}(a_t)\|_{\mathcal{A}_{p+n}}.$$
 The last two equations together with \eqref{clnorm}} imply 
\begin{equation}\label{Brackest}
\|P_{V_{t_n}^{c}}...P_{V_{t_1}^{c}}A_t^c\| \leq
\frac{(p+n-1)!}{(p-1)!}\left(\frac{2\|v\|_{\infty}}{\hbar}\right)^n
\|a\|_{\mathcal{A}_{p}}.
\end{equation}
Equations  \eqref{Vo}, \eqref{2nbnd},\eqref{atn}, \eqref{Brackest}
\DETAILS{, we obtain
$$||A_{t,n}^c|| \leq \frac{1}{n!}\frac{(p+n-1)!}{(p-1)!}{(\frac{2t
||v||_{\infty}}{\hbar})}^{n}  \|a\|_{\mathcal{A}_{p}}$$ 
%
This, due to \eqref{2nbnd}} and the definition $
\tau:=\frac{\hbar}{8\|v\|_{\infty}}$ give \eqref{Voltera}, which
completes the proof of the proposition.
\end{proof}

\section{Hartree von Neumann approximation for $t \le \tau$}\label{secCtau}
 In this section we estimate the difference between the quantum $N$-body
average $Tr( A(t) \rho^{\otimes N}) 
$, where $A(t):=e^{\frac{iH_N t}{\hbar}}A e^{-\frac{iH_N
t}{\hbar}}$, and the classical evolution
 $U_t(A^c(\rho))$, where $\rho$
satisfies $Tr \rho=1.$
\begin{proposition}\label{smalltime}
For all $A \in \mathcal{A}_{N, p}$ and for all $t \leq \tau$ we have
that
\begin{equation} \label{tau_estimate}
|Tr( A(t) \rho^{\otimes N})-U_t(A^c(\rho))|\leq 2^{p+1}\frac{
p+1}{N}
\frac{t}{\tau}\|a\|_{\mathcal{A}_{p}}.
\end{equation}
\end{proposition}
\begin{proof} The proof of this proposition uses the following auxiliary lemma:
\begin{lemma}\label{multit}
For any $p$-particle observable $A$ and any $n \leq N-p$ we have
that
\begin{equation}\label{multitree}
(T_{t_n}T_{t_{n-1}}...T_{t_1}(A_t))^c= \frac{(N-p)!}{(N-p-n)!N^n}
P_{V_{t_n}^{c}}...P_{V_{t_1}^{c}}A_t^c.
\end{equation}
\end{lemma}
\begin{proof} Let $A=\phi_p(a)$ with $a \in \mathcal{A}_p$. Then
$A_t=\phi_p(a_t)$, where, recall, $a_t=e^{\frac{i H_p^0 t}{\hbar}} a
e^{-\frac{i H_p^0 t}{\hbar}}$. 
Now, using the facts that
\begin{equation}\label{multipletree}
T_{t_n}...T_{t_1} \phi_p (a)=\frac{(N-p)!}{(N-p-n)!N^n}
\phi_{p+n}(X_{p+n-1,t_n}X_{p+n-2,t_{n-1}}...X_{p,t_1}(a)),
\end{equation}
which follows from equation  \eqref{newtree}, that $P_S
\rho^{\otimes N}=\rho^{\otimes N}$ and that
$Tr(\phi_p(a)\rho^{\otimes N})= Tr(a\rho^{\otimes p})$, we find 
$$(T_{t_n}...T_{t_1}(A_t))^c(\rho) 
=\frac{(N-p)!}{(N-p-n)!N^n} Tr ( X_{p+n-1,t_n}...X_{p,t_1}(a_t)
\rho^{\otimes p+n}).$$
Now, equation \eqref{multitree} follows from the last equation and
equations \eqref{multibrac}, \eqref{clobs} and $A_t=\phi_p(a_t)$. 
\end{proof}
Now, equations \eqref{multitree} and \eqref{atn} imply that for any
$n \leq N-p$
 \begin{equation}\label{connect}
\int_{\Delta^t_{n}}d^{n}r\  (T_{r_n}...T_{r_1}(A_t))^c=
\frac{(N-p)!}{(N-p-n)!N^n} A_{t,n}^c.
\end{equation}
This, together with  \eqref{treezero}, \eqref{GammaH}, \eqref{clobs}
and \eqref{Hartreeflow}, yields that
\begin{equation}\label{GammatHHartree}
Tr (\Gamma_t^H(A_t)\rho^{\otimes N})-U_t(A^c(\rho))$$
$$=-\sum_{n=1}^{N-p}\left(1-\frac{(N-p)!}{(N-p-n)!N^n}\right)
A_{t,n}^c(\rho)-\sum_{n=N-p+1}^{\infty} A_{t,n}^c(\rho),
\end{equation}
 which
together with \eqref{clnorm}, $ Tr \rho=1$
and \eqref{Voltera} gives for $t \leq \tau$
 \begin{equation}\label{ghartree}|Tr(\Gamma_t^H(A_t)\rho^{\otimes
N})-U_t(A^c(\rho))| \leq \frac{t}{\tau} S_{p,
N}\|a\|_{\mathcal{A}_{p}},
 \end{equation}
where $S_{p, N}:=
\left[\sum_{n=1}^{N-p}\left(1-\frac{(N-p)!}{(N-p-n)!N^n}\right)\
\frac{1}{2^n}+\sum_{n=N-p+1}^{\infty} \frac{1}{2^n}\right]2^{p-1}.$
We transform
 \begin{equation}\label{S}
S_{p, N}:= \left(1-\sum_{n=1}^{N-p} \frac{(N-p)!}{(N-p-n)!
N^n}\frac{1}{2^n}\right) 2^{p-1}.
 \end{equation}
The following inequality is proven in Appendix \ref{numberslemmas}:
\begin{equation}\label{numberseq2}
1-\sum_{n=1}^{N-p} \frac{(N-p)!}{(N-p-n)! N^n}\frac{1}{2^n} \leq
\frac{2(p+1)}{N}.
\end{equation}
Equations \eqref{S} and \eqref{numberseq2} imply that $S_{p, N} \le
2^{p}\frac{p+1}{N}$, which together with \eqref{ghartree} and
\eqref{flowconv} implies that for $t \leq \tau$
 \begin{equation} \label{92}
|Tr ( \Gamma_t(A_t)\rho^{\otimes N})-U_t(A^c(\rho))|\leq
[\frac{p}{N}2^{p-2}+2^p \frac{p+1}{N}]\|a\|_{\mathcal{A}_{p}}.
 \end{equation}
Recall the notation $A(t)=e^{\frac{i H_N t}{\hbar}}Ae^{-\frac{i H_N
t}{\hbar}}$. Due to the equations \eqref{subtractfree} and
\eqref{free}, we have that
$A(t)=\Gamma_t(A_t).$
 This, together with \eqref{92}, implies Proposition
\ref{smalltime}. \end{proof}

\section{Hartree approximation for arbitrary $t$}\label{secC}
Now we prove our main result, Theorem \ref{main}. In what follows $A
=\phi_p(a)$ is a $p$-particle observable and
$\alpha_t(A)=A(t)=e^{\frac{i H_N t}{\hbar}}Ae^{-\frac{i H_N
t}{\hbar}}$.
We proceed by induction. Equations \eqref{tau_estimate} and
$A(t)=\alpha_t(A)$ imply that
\begin{equation}\label{1tau}
|Tr ( \alpha_\tau(A) \rho^{\otimes N})-U_\tau(A^c(\rho))| \le
\frac{2^{p+1} (p+1)}{N}\|a\|_{\mathcal{A}_{p}}.
\end{equation}
Let $L_k = \frac{L_0}{(k-1)!}$ so that $L_{k}=k L_{k+1}$. We assume
that for any $A =\phi_p(a) \in \mathcal{A}_{N, p}$ and for some $k
\ge 1$
\begin{equation}\label{ktau}
|Tr ( \alpha_{k\tau}(A) \rho^{\otimes N})-U_{k\tau}(A^c)| \le
R_{p,k}\|a\|_{\mathcal{A}_{p}},
\end{equation}
where
\begin{equation}
R_{p,k}=2^{kp}\left(2^{\sum_{r=1}^{k}rL_r}
\frac{p}{N}+2^{-L_k}\right),
%
\end{equation}
and prove it for $k+1$.
For $k=1$, \eqref{ktau} follows from \eqref{1tau},  since $L_1 \ge
3$ 
 and $p \ge 1$.

We begin with some preliminary inequalities. Let $A =\phi_p(a)$  and
\begin{equation}\label{amt}
A_n(t):=\int_{\Delta^t_{n}}d^{n}t\  T_{t_n}...T_{t_1}A_t.
\end{equation}
Since $A_t=\phi_p(a_t)$, where $a_t=e^{\frac{i H_p^0
t}{\hbar}}ae^{-\frac{i H_p^0 t}{\hbar}}$ is a $p$-observable, we have by
\eqref{multipletree} that $A_n(t):=\phi_{p+n}(a_n(t))$ with
\begin{equation}\label{pplusn}
 a_n(t)=\frac{(N-p)!}{(N-p-n)!N^n}\int_{\Delta^t_{n}}d^{n}t\
 X_{p+n-1,t_n}X_{p+n-2,t_{n-1}}...X_{p,t_1}(a_t),
\end{equation}
which, together with \eqref{mtaest}, \eqref{Vo}, \eqref{2nbnd} and
\eqref{N-pbnd} and the definition $
\tau:=\frac{\hbar}{8\|v\|_{\infty}}$, gives
\begin{equation}\label{Anest1}\|A_n(t)\|_{\mathcal{A}_{N, p+n}} \leq
\|a_n(t)\|_{A_{p+n}}$$$$\leq
\frac{(p+n-1)!}{(p-1)!n!}\frac{(N-p)!}{(N-p-n)!N^n}
\left(\frac{2\|v\|_{\infty}t}{\hbar}\right)^n
\|a\|_{\mathcal{A}_{p}}\end{equation}
\begin{equation}\label{Anest}
 \leq 2^{p-1-n}\frac{(N-p)!}{(N-p-n)!N^n} \left(\frac{t}{\tau}\right)^n
\|a\|_{\mathcal{A}_{p}}
\end{equation}
\begin{equation}\label{Anest3}
 \leq 2^{p-1-n} \left(\frac{t}{\tau}\right)^n
\|a\|_{\mathcal{A}_{p}}.
\end{equation}
Using that $\alpha_\tau(A)= \Gamma_\tau(A_\tau)$ and using
\eqref{GammaH} with \eqref{flowconv} and \eqref{Anest3}, we obtain
for $L \le N-p$ that
\begin{equation}\label{Nborem}
\|\alpha_\tau(A)-\sum_{n=0}^{L-1} A_n(\tau)\|_{\mathcal{A}_{N}} \leq
2^{p}\left(\frac{p}{4N}+2^{-L}\right)\|a\|_{\mathcal{A}_{p}}.
\end{equation}

Next, we claim that for $L \le N-p$
\begin{equation}\label{Hartrem}
|U_\tau(A^c)-\sum_{n=0}^{L-1} A_n(\tau)^c| \leq
2^{p}\left(\frac{p+1}{N} + 2^{-L}\right)\|a\|_{\mathcal{A}_{p}}.
\end{equation}
Indeed, by \eqref{Hartreeflow}, \eqref{connect} and \eqref{amt} we
have that for $L \le N-p$
\begin{equation}
U_\tau(A^c)-\sum_{n=0}^{L-1}
A_n(\tau)^c=\sum_{n=0}^{L-1}(1-\frac{(N-p)!}{(N-p-n)!N^n})A_{\tau,n}^c
+ \sum_{n=L}^{\infty}A_{\tau,n}^c,
\end{equation}
where $A_0(\tau)=A_{\tau}$. This and equation \eqref{Voltera} imply
\begin{equation}\label{Hartrem'}
|U_\tau(A^c)-\sum_{n=0}^{L-1} A_n(\tau)^c| \leq S_{p,N,
L}\|a\|_{\mathcal{A}_{p}},
\end{equation}
where
$S_{p,N, L}:=
\sum_{n=0}^{L-1}\left(1 -\frac{(N-p)!}{(N-p-n)!N^n}\right)2^{p-1-n}+
\sum_{n=L}^{\infty} 2^{p-n-1}.$ Proceeding as in Eqns \eqref{S} and
\eqref{numberseq2}, we obtain $S_{p,N, L}\leq
\frac{p+1}{N}2^{p} +2^{p-L}.$ This inequality together with
\eqref{Hartrem'} gives \eqref{Hartrem}.

Now we prove \eqref{ktau} for $k+1$ (assuming it for $k$). In what follows we 
use the notation $\langle \cdot \rangle = Tr ( \cdot \rho^{\otimes
N})$. Let $s=k\tau$. We have by \eqref{Nborem} and the linearity and
unitarity of $\alpha_s$ that for $L_{k+1} \le N-p$
$$
|\langle\alpha_{s}(\alpha_\tau(A))\rangle - \langle
\alpha_s(\sum_{n=0}^{L_{k+1}-1} A_n(\tau))\rangle| \le
2^{p}(\frac{ p}{4 N}+2^{-L_{k+1}})\|a\|_{\mathcal{A}_{p}}.$$ Next,
using this inequality, using that $A_n(\tau)$ are $(p+n)$-particle
observables (which follows from equation \eqref{pplusn}) and using
\eqref{ktau} and \eqref{Anest3}, we obtain
$$|\langle\alpha_{s}(\alpha_\tau(A))\rangle -
U_s(\sum_{n=0}^{L_{k+1}-1} A_n(\tau)^c)|$$
$$\le \sum_{n=0}^{L_{k+1}-1} R_{p+n,k}\|a_n(\tau)\|_{\mathcal{A}_{
p+n}}+
2^{p}(\frac{p}{4N}+2^{-L_{k+1}})\|a\|_{\mathcal{A}_{p}}$$
\begin{equation}\label{plustau}
\le 2^{p}\left(\sum_{n=0}^{L_{k+1}-1} R_{p+n,k}2^{-1-n}+
\frac{p}{4N}+2^{-L_{k+1}}\right)\|a\|_{\mathcal{A}_{p}}.
\end{equation}
Equations \eqref{Hartrem}, \eqref{plustau} and $\langle
\alpha_{\tau+s}(A) \rangle
=\langle\alpha_{s}(\alpha_\tau(A))\rangle$
imply
\begin{equation}\label{prelimest}|\langle \alpha_{\tau+s}(A)
\rangle 
-U_s(U_\tau(A^c))| \le T_{p,N}\|a\|_{\mathcal{A}_{p}},
\end{equation}
where
$$T_{p,N}:=2^{p}\left(\sum_{n=0}^{L_{k+1}-1} R_{p+n,k}2^{-1-n}+2^{-L_{k+1}+1}+
2\frac{p+1}{N}\right).$$

We claim that
\begin{equation}\label{remstep}
T_{p,N}\leq R_{p,k+1}.
\end{equation}
Indeed,
$$\sum_{n=0}^{L_{k+1}-1} R_{p+n,k}2^{p-1-n} =\sum_{n=0}^{L_{k+1}-1}
2^{(k+1)p+(k-1)n-1}(2^{\sum_{r=1}^{k}rL_r}
\frac{p+n}{N}+2^{-L_k})$$$$\leq \sum_{n=0}^{L_{k+1}-1}
2^{(k+1)p+(k-1)n-1}(2^{\sum_{r=1}^{k}rL_r}
\frac{p+L_{k+1}}{N}+2^{-L_k})$$
$$\le
2^{(k+1)p+(k-1)L_{k+1}-1}(2^{\sum_{r=1}^{k}rL_r}
\frac{p+L_{k+1}}{N}+2^{-L_k}).$$
 Since $L_k = kL_{k+1}$, we find
$$\sum_{n=0}^{L_{k+1}-1} R_{p+n,k}2^{p-1-n} \le
2^{(k+1)p-1}(2^{\sum_{r=1}^{k+1}rL_r-2L_{k+1}}
\frac{p+L_{k+1}}{N}+2^{-L_{k+1}})$$
$$ \le 2^{(k+1)p-1}(2^{\sum_{r=1}^{k+1}rL_r-L_{k+1}}
\frac{p}{N}+2^{-L_{k+1}}).$$
%
%
This inequality, the definition of $T_{p,N}$ and elementary bounds
imply the estimate \eqref{remstep}, provided $k \ge 2$.
%

Since $s=k\tau$ and since $U_s(U_{\tau}(A^c))=U_{s+\tau}(A^c)$, Eqns
\eqref{prelimest} and \eqref{remstep} imply equation \eqref{ktau}
with $k \rightarrow k+1$. Thus \eqref{ktau} is shown by induction.
Take $k-1$ to be the integer part of $\frac{t}{\tau}$, i.e.
$k-1=[\frac{t}{\tau}]$, and let $\tau':= t/k.$ 
Then $\tau' \le \tau$ and we proceed as above but with $\tau$
replaced by $\tau'$ 
to prove \eqref{ktau} with $\tau$ replaced by $\tau'$. 
Next, using that $\sum_{r=1}^{k}rL_r
=\sum_{r=1}^{k}r\frac{L_0}{(r-1)!} \le 2 e L_0$ and taking
$L_0 = \frac{\ln_2 N}{4e}$ in \eqref{ktau} we arrive at
\begin{equation}\label{arbtest}
|Tr ( \alpha_t(A)\rho^{\otimes N})-U_t(A^c(\rho))|\le
2^{([\frac{t}{\tau}]+1)p}\left(p N^{-1/2}+N^{\frac{-1}{4e
[\frac{t}{\tau}+1]!}}\right)\|a\|_{\mathcal{A}_{p}}.
\end{equation}
Now, by the definition of $\alpha_t(A)$  and $\rho_N$ we have that $
Tr ( \alpha_t(A)\rho^{\otimes N})=Tr ( A\rho_{ N})$, and therefore
\eqref{arbtest} implies \eqref{2star}. Theorem \ref{main} is proven.

\appendix

\section{Appendix: Hartree-von Neumann equation \eqref{Hartree-von Neumann} }\label{app:HvN}
In this section we scketch proofs of some of key properties of the Hartree-von Neumann equation \eqref{Hartree-von Neumann}. For works on the related Hartree-Fock equation see \cite{BDF, C, CG}. Let $\mathbb{I}_1$ denote the
space of all positive, trace class operators on $L^2(\mathbb{R}^3)$.
 \begin{theorem}\label{main}
Assume that $v$ is bounded. Then the Hartree-von Neumann equation \eqref{Hartree-von Neumann} is globally well-posed on $\mathbb{I}_1$ and the trace and the energy are conserved.
 \end{theorem}
\textit{Sketch of Proof.} We will display the $t-$dependence as a subindex. Let $\mathbb{J}_1$ denote the
space of all trace class operators on $L^2(\mathbb{R}^3)$. Using the Duhamel formula we rewrite \eqref{Hartree-von Neumann} as the fixed-point problem $\rho =F(\rho)$ on  $C([0, T], \mathbb{J}_1)$. Here $T$ will be chosen later and
  \begin{equation}
F(\rho)_t := \sigma_t(\rho_0)+i \int_0^t \sigma_{t-s}((v*n_{\rho_s}) \rho_s) ds.
  \end{equation}
where $\sigma_t(\gamma)=e^{\frac{i h t}{\hbar}}\gamma e^{-\frac{i h
t}{\hbar}}$. Denote the trace norm by $\|  \cdotp \|_1$ and let  $\|  \rho \|_T := \sup_{0 \le s \le T}\|  \rho_s \|_1 $ be the norm on the space $C([0, T], \mathbb{J}_1)$. Recall that  $\|  f \|_\infty$ denotes the $L^\infty$-norm of a function $f$. Let $v_x(y):=v(x-y)$. Using that 
$|(v*n_\rho)(x)|= |Tr (v_x\rho)|  \le \|  v \|_\infty\|\rho\|_1$, we obtain
  \begin{equation}
\|  \sigma_{t-s}((v*n_{\rho_s}) \rho_s) \|_1 \le \|  v*n_{\rho_s} \|_\infty \| \rho_s \|_1 \le  \|  v \|_\infty \| \rho \|_T^2.
  \end{equation}
This estimate shows that $F$ maps any ball in $C([0, T], \mathbb{I}_1)$ of radius $R \ge 2$ into itself, provided $T \le 1/ \|  v \|_\infty R^2$. Similarly, one shows that $F$ is a contraction on such a ball, if  $T \le 1/2 \|  v \|_\infty R$. Hence our fixed point equation in any $B_R,\ R \ge 2,$ has a unique solution for $T = 1/ 2\|  v \|_\infty R^2$. This solution solves also the original intial value problem \eqref{Hartree-von Neumann}.

Since $\rho_t$ and $\rho_t^*$ satisfy the same equation \eqref{Hartree-von Neumann} with the same initial condition $\rho_0=\rho^*_0$, we conclude by uniqueness that $\rho_t=\rho_t^*$. Since the trace of a commutator vanishes, one has that the trace of $\rho_t$ is independent of $t$.

We show that the eigenvalues of $\rho_t$ are independent of $t$. We denote by $\lambda_i$ and $\phi_i$ the eigenvalues and the corresponding eigenfunctions of $\rho_t$ and compute
  \begin{equation}
\partial_t \lambda_k = \partial_t \langle\phi_k, \rho_t \phi_k\rangle =  \langle\partial_t\phi_k, \rho \phi_k\rangle +\langle\phi_k, \dfrac{1}{i}[h_\rho,\rho_t] \phi_k\rangle + \langle\phi_k, \rho_t \partial_t \phi_k\rangle.
  \end{equation}
Since $\rho$ is self-adjoint, $\rho_t\phi_k=\lambda_k\phi_k$  and $\langle\phi_k, [h_{\rho_t},\rho_t] \phi_k\rangle=0$, this gives
 \begin{equation}
\partial_t \lambda_k =  \lambda_k\langle\partial_t\phi_k,  \phi_k\rangle  + \lambda_k\langle\phi_k,  \partial_t \phi_k\rangle =\lambda_k \partial_t \langle\phi_k,  \phi_k\rangle =0.
  \end{equation}

Since the eigenvalues of $\rho_t$ are independent of $t$ and since $\rho_0$ is non-negative,  $\rho_t$ is non-negative as well. Hence  $\| \rho_t \|_1 = Tr \rho_t $ and is independent of $t$. Therefore the local well-psedness of \eqref{Hartree-von Neumann} can be extended to the global one. 

The conservation of energy is proven in a standard way. \qed
    
\section{Appendix: Jacobi identity for \eqref{poisson} }\label{app:Jacobi}
 \begin{lemma} The Poisson bracket defined in \eqref{poisson} satisfies the Jacobi
identity:
  \begin{equation}
\{\{A,B\},C\}+\{\{C,A\},B\}+\{\{B,C\},A\}=0.
  \end{equation}
\end{lemma}
  \begin{proof} In what follows we omit the argument $\rho$ and the superindex $c$.  Denote  $A'=\partial_{\rho}A$. Recall the definition, $\{A,B\}= -\frac{i}{\hbar}Tr
\left(A'\rho
 B'-B'\rho
A'\right)$. We have
  \begin{equation}
\{A,B\}=\frac{i}{\hbar}Tr([A',B']\rho).
  \end{equation}
We further denote $A''=\partial_{\rho}^2A$. We think about $A''$ as
an operator on a two-particle space, or a two-particle operator.
If we let $K_1=K \otimes I$ and $ K_2=I
\otimes K$, then we have (omitting from now on the factor $\frac{i}{\hbar}$) \begin{equation}  \label{PB1}
\{A,B\}'=Tr_2(([A'',B_2']+[A_2',B''])\rho_2)+[A',B'],
\end{equation}
where $Tr_2$ denotes the partial trace with respect to
the second coordinate. 
The last two relations give 
\begin{equation}
\{\{A,B\},C\}=Tr([\{A,B\}',C']\rho)$$$$=Tr_1([Tr_2([A'',B_2']\rho_2+[A_2',B'']\rho_2),C_{1}']\rho_1 )
+Tr([[A',B'],C']\rho ).
\end{equation}
Ler $Tr^{\otimes 2}$ denote the trace over the two-particle space. We have furthermore
$$Tr_1([Tr_2([A'',B_2']\rho_2),C_{1}']\rho_{1})$$
$$=Tr^{\otimes 2}([[A'',B_2'],C_1'] \rho_1 \rho_2)$$
$$=Tr^{\otimes 2}([[A'',B_1'],C_2'] \rho_1 \rho_2).$$
and similarly for the second term on the r.h.s. of \eqref{PB1}.
Let $R:=\rho \otimes \rho$ and denote $Tr^{\otimes 2}(K \rho \otimes
\rho)=Tr_R(K)$.
  Then we have,
  $$\{\{A,B\},C\}=Tr_R([[A'',B_1'],C_2']+[[A_2',B''],C'_1])+Tr([[A',B'],C']\rho).$$
Since the commutator of operators satisfies the Jacobi identity, the last equation implies that
  \begin{equation}
\{\{A,B\},C\}+\{\{C,A\},B\}+\{\{B,C\},A\}$$$$=Tr_R([[A'',B_1'],C_2']+[[B'',C_1'],A_2']+[[C'',A_1'],B_2']$$$$+[[A_2',B''],C_1']+[[B_2',C''],A_1']+[[C_2',A''],B_1'])=0.
  \end{equation}
  \end{proof}

\section{Appendix. Proof of Eqn \eqref{numberseq2}}\label{numberslemmas}
\begin{lemma}\label{numberslemma2}
\begin{equation}\label{numberseq2'}
1-\sum_{n=1}^{N-p} \frac{(N-p)!}{(N-p-n)! N^n}\frac{1}{2^n} \leq
\frac{2(p+1)}{N}.
\end{equation}
\end{lemma}
\begin{proof}
We derive this inequality from the following lemma:
\begin{lemma}
\begin{equation}\label{beta}
1-\sum_{n=1}^{N-p} \frac{(N-p)!}{(N-p-n)!
N^n}\left(\frac{1}{2}\right)^n=\sum_{n=0}^{N-p}\frac{(N-p)!}{(N-p-n)!N^n}\frac{p+n}{N}\left(\frac{1}{2}\right)^n
\end{equation}
\end{lemma}
\begin{proof}
Let
\begin{equation}\label{beta1}
B=\sum_{n=1}^{N-p} \frac{(N-p)!}{(N-p-n)!
N^n}\left(\frac{1}{2}\right)^n
\end{equation}
Then,
\begin{equation}\label{beta2}
2B=\sum_{n=1}^{N-p} \frac{(N-p)!}{(N-p-n)!
N^n}\left(\frac{1}{2}\right)^{n-1}=\sum_{n=0}^{N-p-1}
\frac{(N-p)!}{(N-p-n-1)! N^{n+1}}\left(\frac{1}{2}\right)^n
\end{equation}
Subtracting equation \eqref{beta1} from \eqref{beta2}, we obtain
that
$$B=\frac{N-p}{N}-\frac{(N-p)!}{N^{N-p}} \left(\frac{1}{2}\right)^{N-p}-\sum_{n=1}^{N-p-1}\frac{(N-p)!}{(N-p-n)!N^n}(1-\frac{N-p-n}{N})\left(\frac{1}{2}\right)^n $$
$$=1-\frac{p}{N}-\frac{(N-p)!}{N^{N-p}} \left(\frac{1}{2}\right)^{N-p}-\sum_{n=1}^{N-p-1}\frac{(N-p)!}{(N-p-n)!N^n} \frac{p+n}{N} \left(\frac{1}{2}\right)^n$$
$$=1-\sum_{n=0}^{N-p}\frac{(N-p)!}{(N-p-n)!N^n} \frac{p+n}{N} \left(\frac{1}{2}\right)^n,$$
which gives immediately equation \eqref{beta}.
\end{proof}
On the other hand we have that
$$\sum_{n=0}^{N-p}\frac{(N-p)!}{(N-p-n)!N^n}\frac{p+n}{N}\left(\frac{1}{2}\right)^n$$
$$\leq \sum_{n=0}^{N-p}\frac{p+n}{N}\left(\frac{1}{2}\right)^n \stackrel{\eqref{numberseq1}}{\leq} \frac{2(p+1)}{N},$$
where in the last step we used equation \eqref{numberseq1} with $p
\rightarrow p+1$.
 This, together with equation \eqref{beta} gives Lemma
\ref{numberslemma2}.
\end{proof}


\begin{thebibliography}{DGMS}
\bibitem[AW]{AW} Abou-Salem W.K.: \emph{A Remark on the mean field dynamics of many body bosonic systems with random interactions in a random
potential}. Lett. Math. Phys. 84:231-243 (2008).
\bibitem[AN]{AN} Ammari Z., Nier  F.: \emph{Mean field limit for bosons
and infinite dimensional phase space analysis}. Annales Henri
Poincar\'e, vol. 9, issue 8, pp. 1503-1574 (2008).
\bibitem[An]{An} Anapolitanos I.: \emph{Mean field dynamics of the quantum many body problem}. PhD Thesis, 2009, In preparation.
\bibitem[BEGMY]{BEGMY} Bardos C., Erd\"os L., Golse F., Mauser N., Yau
H-T.: \emph{Derivation of the Schr\"odinger-Poisson equation from
the quantum N-body problem}. C.R. Acad. Sci. Paris, Ser. I334 (2002)
515-520.
\bibitem[BGGM]{BGGM} Bardos C.,  Golse F., Gittlieb A.D., Mauser
N.: \emph{Mean field dynamics and the time-dependent Hartree-Fock
equation}. J. Math. Pures Appl. 82 (2003) 665-683.
\bibitem[BDF]{BDF} Bove, A., Da Prato, G., Fano, G.:\emph{An
existence proof for the Hartree-Fock time-dependent problem with
bounded two-body interaction.} Commun. Math. Phys. 37, 183-191
(1974).
\bibitem[C]{C} Chadam, J. M. :  \emph{The time-dependent Hartree-Fock equations with Coulomb two-body interaction}.
Commun. Math. Phys., 46 (1976), pp. 99–104.
\bibitem[CG]{CG} Chadam, J. M. , Glassey,  R. T. : \emph{Global existence of solutions to the Cauchy problem for
time-dependent Hartree equations}. J. Math. Phys. 16, 1122 (1975).
\bibitem[ES]{ES} Elgart A., Schlein B.: \emph{Mean field dynamics of boson
stars}. Comm. Pure and Applied Math. 60 , no. 4, 500-545 (2007).
\bibitem[ESY]{ESY} Elgart A., Erd\"os L., Schlein B.,  Yau H-T.:  \emph{Gross-Pitaevskii equation as the mean field limit of weakly coupled
Bosons}. Arch. Rat. Mech. Anal. 179, no. 2, 265-283, (2006).
\bibitem[EESY]{EESY} Elgart A., Erd\"os L., Schlein B.,  Yau H-T.: \emph{Nonlinear Hartree equation as the mean field limit of weakly coupled
fermions}. J. Math. Pure Appl., 83, 1241 (2004).
\bibitem[ErS]{ErS} Erd\"os L, Schlein B.: \emph{Quantum dynamics with mean field interactions: a New Approach}.
arXiv:0804.3774v1 (2008).
\bibitem[ErY]{ErY} Erd\"os L., Yau H-T.: \emph{Linear Boltzmann equation as the weak coupling limit of a random Schr\"odinger equation}. Commun. Pure Appl. Math. 53
(6), 667-735 (2000).
\bibitem[FGS]{FGS} Fr\"ohlich J., Graffi S., Schwartz S.: \emph{Mean field and classical limit of many body Schr\"odinger dynamics
for Bosons}. Commun. Math. Phys. 271, 681-697 (2007).
\bibitem[FKP]{FKP} Fr\"ohlich J., Knowles A., Pizzo A.: \emph{Atomism and
quantization}. J. Phys. A: Math. Theor. 40 3033-3045 (2007).
\bibitem[FKS]{FKS} Fr\"ohlich J., Knowles A., Schwartz S.: \emph{On the mean-field limit of bosons with Coulomb two body
interaction}. Preprint ArXiv:0805.4299 (2008).
\bibitem[GV]{GV} Ginibre J., Velo G.: \emph{The classical field limit for nonrelativistic bosons II.
Asymptotic expansions for general potentials.} Annales de l'institut
Henri Poincar\'e    (A) Physique th\'eorique, 33 no. 4 , p. 363-394
(1980).
\bibitem[GM]{GM} Grillakis M.G. and Margetis D.: \emph{Aprioi estimates for many-body Hamiltonian evolution of interacting boson
system}. Journal of Hyperbolic Differential Equations Vol.5, No. 4
857-883 (2008).
\bibitem[GMM]{GMM} Grillakis M.G, Machedon M., Margetis D.:\emph{Second order corrections to mean field interaction for weakly interacting
bosons,I}. Arxiv:09040158v1 (2009).
\bibitem[He]{He} Hepp K.: \emph{The classical limit for quantum mechanical correlation
functions}. Commun. math. Phys. 35, 265-277 (1974).
\bibitem[LS]{LS} Lieb E. Seiringer R.: \emph{Proof of Bose-Einstein condensation for dilute trapped
gases}. Phys. Rev. Lett. 88, 170409 (2002).
\bibitem[LSY]{LSY} Lieb E., Seiringer R., Yngvason J. \emph{One-Dimensional bosons in three dimensional traps}. Phys. Rev. Lett. 91,150401
(2003).
\bibitem[KSS]{KSS} Kirkpatrick K., Schlein B., Staffilani G.: \emph{Derivation of the two dimensional nonlinear Sch\"odinger equation from many body quantum dynamics}. ArXiv:0808.0505
(2008).
\bibitem[KF]{KF}  Knowles A., Fr\"ohlich J.: \emph{A microscopic derivation of the time-dependent Hartree-Fock equation with Coulomb two-body interaction}. ArXiv:0810.4282
(2008).
\bibitem[KM]{KM} Klainerman S., Machedon M.: \emph{On the uniqueness of solutions to the Gross-Pitaevskii
hierarchy}. Comm. Math. Phys., Volume 279, Number 1, 2008.
\bibitem[ReSi]{ReSi} Reed M., Simon B.: Methods of Modern
Mathematical Physics, IV: Analysis of Operators. Academic press
(1978).
\bibitem[RS]{RS} Rodnianski I., Schlein B.: \emph{Quantum fluctuations and rate of convergence towards mean field
dynamics}.  To appear in Comm. Math. Phys., Preprint
ArXiv:0711.3087.
\bibitem[S]{S} Spohn H.: \emph{Kinetic equations from Hamiltonian
dynamics}. Rev. Mod. Phys 52 (1980), no. 3, 569-615.
\end{thebibliography}
\end{document}